\documentclass[12pt, a4paper]{article}
\usepackage[utf8]{inputenc}

\usepackage{amsmath}
\usepackage{amssymb}
\usepackage{amsthm}
\usepackage{authblk}
\usepackage{braket}
\usepackage{framed}
\usepackage{algorithm2e}
\usepackage{xcolor}

\usepackage{tikz}
\usetikzlibrary{positioning}
\usetikzlibrary{decorations.pathreplacing}

\newcommand{\smgd}{\textsf{SMGD}\xspace}
\newcommand{\HH}{\mathcal{H}\xspace}

\DeclareMathOperator{\shiftm}{ShiftedMatch}
\DeclareMathOperator{\fullm}{FullMatch}
\DeclareMathOperator{\forwm}{ForwardMatch}
\DeclareMathOperator{\backm}{BackwardMatch}

\newtheorem{theorem}{Theorem}

\newtheorem{lemma}[theorem]{Lemma}
\newtheorem{problem}{Problem}

\title{Quantum Pattern Matching in Generalised Degenerate Strings}
\date{}

\author[1]{Massimo Equi}
\author[2]{{Md Rabiul Islam} Khan}
\author[2]{Veli M\"akinen}

\affil[1]{Aalto University, Finland}
\affil[2]{University of Helsinki, Finland}

\begin{document}
\maketitle
\begin{abstract}
A \emph{degenerate string} is a sequence of sets of characters. A \emph{generalized degenerate (GD) string} extends this notion to the sequence of sets of strings, where strings of the same set are of equal length. Finding an exact match for a pattern string inside a GD string can be done in $O(mn+N)$ time (Ascone et al., WABI 2024), where $m$ is the pattern length, $n$ is the number of strings and $N$ the total length of strings constituting the GD string. This is the best classical algorithm achieved so far, and no matching lower bound, neither unconditional nor conditional, has been shown. We make progress on this problem proposing a quantum algorithm that achieves running time $\tilde{O}(\sqrt{mnN})$, thus beating the current best classical solution. To the best of our knowledge, this is the first quantum algorithm proposed in the context of GD strings. We present our results starting from the framework of classical parallel computing, which we believe makes them intuitive to understand and possibly easy to generalize to other similar structures.
\end{abstract}


\section{Introduction}
\label{sec:intro}
The exact string matching problem is to decide if a pattern string $P$ of length $m$ appears as a substring of a text string $T$ of length $n$. This problem can be solved in $O(m+n)$ time \cite{KMP77} under classical models of computation. The first quantum algorithm for a string matching problem was given by Ramesh and Vinay~\cite{RV03}, whose solution could find an exact match in time $\tilde{O}(\sqrt{n}+\sqrt{m})$, where $\tilde{O}$ hides logarithmic factors. Since then, other quantum algorithms have been proposed that improve the classical computational complexity of many different string problems such as \emph{longest common substring}~\cite{LeGallSeddighin23,AkmalJin23}, \emph{longest palindrome substring}~\cite{LeGallSeddighin23}, \emph{minimal string rotation}~\cite{WangYing24,AkmalJin23}, \emph{longest square substring}~\cite{AkmalJin23}, \emph{longest common subsequence}~\cite{JinNogler24}, \emph{edit distance}~\cite{BoroujeniEGHS21,GibneyJKT24}, and \emph{multiple string matching}~\cite{QAC25}.

Nevertheless, not much have been explored in regimes where the input is a more general structure than just a string. For matching strings in labeled graphs, two different quantum approaches have been developed \cite{DGT22,EGM23}. One is tailored to non-sparse graphs, where it gives a better bound than the best algorithm in the classical setting, although no matching classical lower bound is known. The other is restricted to \emph{level-DAGs}, where the input graph is assumed to be a sequence of sets of nodes $V$, with edges $E$ only defined between nodes of consecutive sets, and nodes have character labels. In this setting, the proposed quantum algorithm achieves $O(|E|\sqrt{m})$ running time, improving over the best classical quadratic algorithm and thus overcoming the classical quadratic conditional lower-bound. Both of these approaches only work in special cases, and it is currently open how to cover the general case, and how to prove matching quantum lower bounds.

\begin{figure}[ht]
\centering
\resizebox{\textwidth}{!}{%
\begin{tikzpicture}[
    nuc/.style={font=\ttfamily},
    emph/.style={font=\ttfamily\bfseries\Large, underline}
]

\node (S0) {
\begin{tikzpicture}
\node[nuc] at (0,0)   {A};
\node[nuc] at (0.6,0) {C};
\node[nuc] at (1.2,0) {G};

\node[nuc] at (0,-0.6)   {T};
\node[nuc] at (0.6,-0.6) {A};
\node[nuc] at (1.2,-0.6) {A};

\node[nuc] at (0,-1.2)   {C};
\node[nuc] at (0.6,-1.2) {G};
\node[nuc] at (1.2,-1.2) {T};

\node[nuc] at (0,-1.8)   {G};
\node[nuc] at (0.6,-1.8) {T};
\node[nuc] at (1.2,-1.8) {A};
\end{tikzpicture}
};

\draw[decorate, decoration={brace, mirror, amplitude=6pt}]
  (S0.north west) -- (S0.south west);

\draw[decorate, decoration={brace, amplitude=6pt}]
  (S0.north east) -- (S0.south east);

\node[right=1cm of S0] (S1) {
\begin{tikzpicture}
\node[nuc] at (0,0)   {G};
\node[nuc] at (0.6,0) {A};
\node[nuc] at (1.2,0) {T};
\node[nuc] at (1.8,0) {C};

\node[nuc]  at (0,-0.6)   {C};
\node[nuc]  at (0.6,-0.6) {G};
\node[nuc] at (1.2,-0.6) {\underline{G}};
\node[nuc] at (1.8,-0.6) {\underline{T}};
\end{tikzpicture}
};

\draw[decorate, decoration={brace, mirror, amplitude=6pt}]
  (S1.north west) -- (S1.south west);

\draw[decorate, decoration={brace, amplitude=6pt}]
  (S1.north east) -- (S1.south east);

\node[right=1cm of S1] (S2) {
\begin{tikzpicture}
\node[nuc] at (0,0)   {A};
\node[nuc] at (0.6,0) {C};

\node[nuc] at (0,-0.6)   {\underline{G}};
\node[nuc] at (0.6,-0.6) {\underline{T}};

\node[nuc] at (0,-1.2)   {C};
\node[nuc] at (0.6,-1.2) {A};
\end{tikzpicture}
};

\draw[decorate, decoration={brace, mirror, amplitude=6pt}]
  (S2.north west) -- (S2.south west);

\draw[decorate, decoration={brace, amplitude=6pt}]
  (S2.north east) -- (S2.south east);

\node[right=1cm of S2] (S3) {
\begin{tikzpicture}
\node[nuc] at (0,0)   {\underline{T}};
\node[nuc] at (0.6,0) {\underline{A}};
\node[nuc] at (1.2,0) {\underline{A}};
\node[nuc]  at (1.8,0) {G};
\node[nuc]  at (2.4,0) {T};

\node[nuc] at (0,-0.6)   {A};
\node[nuc] at (0.6,-0.6) {T};
\node[nuc] at (1.2,-0.6) {G};
\node[nuc] at (1.8,-0.6) {C};
\node[nuc] at (2.4,-0.6) {A};
\end{tikzpicture}
};

\draw[decorate, decoration={brace, mirror, amplitude=6pt}]
  (S3.north west) -- (S3.south west);

\draw[decorate, decoration={brace, amplitude=6pt}]
  (S3.north east) -- (S3.south east);

\node[right=1cm of S3] (S4) {
\begin{tikzpicture}
\node[nuc] at (0,0)   {A};
\node[nuc] at (0.6,0) {C};
\node[nuc] at (1.2,0) {G};

\node[nuc] at (0,-0.6)   {T};
\node[nuc] at (0.6,-0.6) {T};
\node[nuc] at (1.2,-0.6) {A};
\end{tikzpicture}
};
\draw[decorate, decoration={brace, mirror, amplitude=6pt}]
  (S4.north west) -- (S4.south west);

\draw[decorate, decoration={brace, amplitude=6pt}]
  (S4.north east) -- (S4.south east);

\end{tikzpicture}
}
\caption{A GD string $T[1..5]$ with $T[1]=\{\mathtt{ACG,TAA,CGT,GTA}\}$, $T[2]=\{\mathtt{GATC,CGGT}\}$, $T[3]=\{\mathtt{AC,GT,CA}\}$, $T[4]=\{\mathtt{TAAGT,ATGCA}\}$, and $T[5]=\{\mathtt{ACG,TTA}\}$. Underlined characters illustrate a match for pattern $\mathtt{GTGTTAA}$.\label{fig:GD}}
\end{figure}

Between matching a string into another string and matching a string into a graph, there is a middle ground of well-studied structures. This is the case for \emph{generalized degenerate strings}. A generalized degenerate (GD) string is a sequence of sets of strings called segments, where strings of the same segment are of equal length, as Figure~\ref{fig:GD} illustrates. This definition can be made more restrictive by imposing constraints on the length of the segments (\emph{degenerate strings} if all segments have length $1$), or more loose by allowing strings in a segment to have different lengths (\emph{elastic degenerate strings}). There are lines of research exploring each one of these variants and more~\cite{IliopoulosMR08,GrossiILPPRRVV17,IliopoulosKP21,EquiNACTM21,RizzoENM24,AlzamelABGIPPR18,AlzamelABGIPPR20,MakinenCENT20}. An extensive summary of these approaches and novel techniques can be found in~\cite{Asconeetal24}. Despite this popularity in the string matching community, no quantum approach has been proposed so far for any of these structures, and we aim to start filling this gap. In this paper, we focus on solving string matching in generalized degenerate strings defined as follows:
\begin{problem}[String Matching in Generalized Degenerate Strings (\smgd)]
\label{problem:smgd}
\item{\textsc{input}:} A GD string $T$ and a pattern string $P$, both over an alphabet $\Sigma$.
\item{\textsc{output}:} \emph{True} if and only if there is at least one occurrence of $P$ in $T$.
\end{problem}
For a formal definition of an \emph{occurrence} of $P$ in $T$, see Section~\ref{subsec:gd-strings}.

\subsection*{Our Results}
We expand the scope of quantum techniques to variants of string matching by proposing the first quantum algorithm for Problem~\ref{problem:smgd}. The solution we propose can be seen as a generalization of the approach of Equi et al.~\cite{EGM23}. Our framework provides a different and more intuitive perspective for designing quantum algorithms for string matching in degenerate strings and graphs. The core idea is to generate a superposition over all shifts of the pattern, where each quantum thread tries to find a match for a different shift. In this light, the main design aspect consists in defining what a single thread is doing. As these can be interpreted as classical threads, possibly with quantum subroutines, this approach should make algorithm design accessible also to non-quantum specialists. To substantiate these claims, in Section~\ref{sec:rephrasing} we discuss how we can rephrase the previous result of Equi et al.~\cite{EGM23} in our framework.

We remark that this generalization is needed to solve quantum string matching in the case of GD strings, as a black-box application of the algorithm of Equi et al.~\cite{EGM23} would not lead to any quantum advantage. To see this, try to represent a GD string as a level-DAG: add edges between all pairs of strings from consecutive segments and split strings to paths of character-labeled nodes. The quantum algorithm of Equi et al.~\cite{EGM23} then applies to solve this GD string pattern matching problem, but the reduction is not efficient. For example, consider a GD string consisting of $N$ characters, with $O(\sqrt{N})$ strings of logarithmic length per segment, and $O(\sqrt{N})$ segments. If we turn such a GD string into a level-DAG, every pair of consecutive segments creates $O(N)$ edges, for a total of $|E|=\Omega(N\sqrt{N})$ edges. This yields an $O(N\sqrt{Nm})$-time quantum algorithm, where $m$ is the length of the pattern. However, the problem can be solved faster in the classical setting in $O(nm+N)$ time~\cite[Theorem 3]{Asconeetal24} for a generalized degenerate string of $n$ segments and $N$ total characters. Taking inspiration from both these quantum~\cite{EGM23} and classical~\cite{Asconeetal24} techniques, we devise a parallel algorithm and then we cast it in the quantum setting, achieving the following.
\begin{theorem}
    \label{thm:main-quantum-algo}
    There exists a quantum algorithm that solves \smgd on a pattern string $P$ of length $m$ and a generalized degenerate string $T$ of $n$ segments and $N$ total characters in $\tilde{O}\left(\sqrt{mnN}\right)$ time, with success probability $p>2/3$. The location of one match can be reported in additional $O(\sqrt{nN})$ time.
\end{theorem}
We remind that the $\tilde{O}$ notation hides factors of complexity $O(n^{o(1)})$, which in our case means any (poly)logarithmic factor.

\subsection*{Technical Overview}

We first describe a classical parallel approach in Section~\ref{sec:main-classical-algo}, and then show how this translates to a quantum algorithm in Section~\ref{sec:main-quantum-algo}. To understand our approach, consider this very simple naive classical sequential solution to \smgd. Given pattern $P$ of length $m$ and GD string $T$ of $n$ segments and $N$ total characters, start from the leftmost position in $T$, called column (formally defined in Section~\ref{sec:preliminaries}), and try to match $P$. Consume characters from $P$ and compare them to $T$ column-by-column (using tries, see Section~\ref{sec:main-classical-algo}) checking whether a character-match can be found or not. Even if there is no character-match, do not stop, but continue consuming one character in $P$ for each column in $T$ until there are no more characters in $P$. At this point, start again from the beginning of $P$, and keep consuming its characters in the same manner until either a full pattern match is found or we reach the end of $T$. If no pattern match is found at this point, restart the procedure, this time starting from the second column of $T$. Keep restarting the procedure, each time shifting the starting position one column to the right, until a match is found or $m$ instances of the procedure have been run. If no match has been found until this point, we can safely stop and report that there is no match, because the subsequent instances of our procedure would start from a column already checked by a previous instance.

The classical parallel algorithm, illustrated in Fig.~\ref{fig:shifted_threads} on page \pageref{fig:shifted_threads}, follows exactly this logic, with the difference that it launches $m$ threads to run the $m$ instances of the above procedure in parallel. In a sense, this is similar to the approach of Ascone et al.~\cite[Theorem 3]{Asconeetal24}, but there the authors use a bit-vector to keep track of the active prefixes between two segments, while here we parallelize over all possible shifts of the pattern against the GD string.

The quantum algorithm replaces the threads with a superposition and, instead of the tries, it finds matches of substrings of the pattern in a segment employing two nested Grover's searches. Using tries in the quantum algorithm can lead to comparable performances when scanning the GD strings, but using Grover's searches avoids any preprocessing, saving an additive $O(N)$ term in the final complexity.

We remark that the techniques in this paper generalize the approach of Equi et al.~\cite{EGM23}, which uses quantum parallelism to simulate a bit-parallel classical algorithm, and thus it is clearly a special case of a multi-threaded algorithm. In the GD-string setting, fully exploiting quantum speed-ups seems infeasible when relying only on a pure bit-parallel abstraction. Conversely, the multi-threaded abstraction where each thread takes care of one shift of the pattern can describe both our algorithm and the algorithm by Equi et al.~\cite{EGM23}.

We also acknowledge a difficulty in finding matching quantum lower bounds for our solution. This is somewhat expected, as current lower bound strategies~\cite{BuhrmanPS21,DGT22} struggle to provide lower bounds for the quantum complexity of finding an exact match for a string in a graph, and that is clearly a harder problem than \smgd. We note however that quantum conditional lower bounds exists for string problems such as edit distance~\cite{BuhrmanPS21}, longest common~\cite{BuhrmanPS21} subsequence and approximate string matching in labeled graphs~\cite{DGT22}, suggesting that beating time complexity $O(n^{1.5})$ is unlikely. For \smgd when $N$, $n$ and $m$ are comparable, our solution features a time complexity of $O(\sqrt{nmN})=O(n^{1.5})$ as well, thus we find it not unreasonable to think that similar lower bounds could exist for this case. 

Finally, we find that, as our algorithm is based on the property of covering the GD string with shifts of the pattern, our techniques could be of independent, even non-quantum, interest.

\section{Preliminaries}
\label{sec:preliminaries}

\subsection{Generalized Degenerate Strings}
\label{subsec:gd-strings}
An \emph{alphabet} $\Sigma$ is a set of \emph{characters}. A sequence $P\in \Sigma^m$ is called a \emph{string} and its length is denoted $m=|P|$. We denote integers $i,i+1,\ldots,j$ as interval $[i..j]$ and represent a string $P$ as an array $P[1..m]$, where $P[i]\in \Sigma$ for $1\leq i\leq m$. String $P[i..j]$ is called a \emph{substring}, string $P[1..i]$ a \emph{prefix}, and string $P[i..m]$ a \emph{suffix} of $P[1..m]$.  

A \emph{generalized degenerate (GD) string} $T[1..n]$ is a sequence $T[1] T[2] \cdots T[n]$ of non-empty sets of fixed length strings, that is, each $T[i] \subseteq \Sigma^{k_i}$, where $k_i$ is a positive integer and $|T[i]|>0$ for all $i$. We denote the \emph{size} of $T$ as the total number of characters $N=\sum_{i=1}^n \sum_{S\in T[i]} |S|$. Moreover, $W=\sum_{i=1}^{n}{k_i}$ is the \emph{width} of $T$. The \emph{language} of $T$ is the set of strings $\{S_1S_2\cdots S_n \mid S_1 \in T[1], S_2 \in T[2], \ldots, S_n \in T[n]\}$.

In this work, we study the problem of \emph{string matching in generalized degenerate strings} (\smgd, Problem~\ref{problem:smgd} in Section~\ref{sec:intro}), which consists in finding a match for a pattern string $P[1..m]$ in a generalized degenerate string $T[1..n]$, where $P$ has a \emph{match} in $T$ if i) $P$ is a substring of any $S \in T[i]$ for any $i$, or ii) there is a sequence of strings $S_i,S_{i+1},\ldots,S_{j-1},S_{j}$ such that $S_i \in T[i],S_{i+1}\in T[i+1],\ldots,S_{j-1} \in T[j-1],S_j \in T[j]$ and $P = S_i[a..k_i] S_{i+1} \cdots S_{j-1} S_{j}[1..b]$ for some integers $a,b>0$. We then can say that a match starts at column $c_1$ and ends at column $c_2$, where $c_1=a+\sum_{x=1}^{i-1}k_x$ and  $c_2=b+\sum_{x=1}^{j-1}k_x=a+m-1$, where $m$ is the length of $P$. For example, in Figure~\ref{fig:GD} $S_2=\mathtt{CGGT}$, $S_3=\mathtt{GT}$, $S_4=\mathtt{TAAGT}$, and $P=\mathtt{GTGTTAA}=S_2[3..4]S_3 S_4[1..3]$, starting at column $6$ and ending at column $12$.

We first propose a classical algorithm that uses a \emph{forward trie} and a \emph{backward trie} for each $T[i]$. The forward trie is a tree on strings in $T[i]$ such that each string $S$ of $T[i]$ corresponds to a distinct leaf $v_S$ and one can follow character-labeled edges from the root to the leaf $v_S$ to spell $S$. The backward trie is the forward trie for the set of strings formed by reversing the strings in $T[i]$. The reverse of string $S$ is $S^r=S[k]S[k-1]\cdots S[1]$, where $k=|S|$.

\subsection{Quantum computation}
In what follows, we introduce our quantum notation, but we assume the reader is familiar with the basic notions of quantum computing as covered in textbooks \cite{NC10QCQI}. In quantum computation, the state of a qubit is described by a vector $\ket{\psi}=\alpha\ket{0}+\beta\ket{1}$, that is a \textit{superposition} of vectors $\ket{0}=(1,0)$ and $\ket{1}=(0,1)$, where $\alpha,\beta \in \mathbb{C}$ are \textit{amplitudes} satisfying $|\alpha|^2+|\beta|^2=1$. Vectors $\ket{0}$ and $\ket{1}$ form the so-called computational basis, which spans a two-dimensional Hilbert space $\HH$ of single-qubit states. The tensor product of two quantum states $\ket{\psi}$ and $\ket{\phi}$ can be written as $\ket{\psi}\otimes\ket{\phi}$ or $\ket{\psi\phi}$ and it is itself a quantum state. In particular, we can take the tensor product of multiple computational basis vectors and obtain the state $\ket{i}=\bigotimes_{k=1}^i\ket{b_k}$, where $b_k$ is the $k$-th bit of the binary representation of $i$. For example, $\ket{5}=\ket{101}=\ket{1}\otimes\ket{0}\otimes\ket{1}$. Thus, the set of vectors $\{\ket{i}\,|\, 0\leq i\leq n-1\}$ forms the computational basis of a $2^n$-dimensional Hilbert space $\HH^{\otimes n}$. The quantum state of multiple qubits can then be represented as a vector in this Hilbert space, and can be written as $\ket{\psi} = \sum_{i=1}^n\alpha_i\ket{i}$, where the square of the amplitudes $\alpha_i \in \mathbb{C}$ is normalized as $\sum_{i=1}^n|\alpha_i|^2=1$. When measuring state $\ket{\psi}$ in the computational basis, the result is $i$ with probability $|\alpha_i|^2$ and the state collapses to $\ket{i}$. A state $\ket{\psi}\in \HH^{\otimes n}$ is called \textit{separable} with respect to the partition $\HH^{\otimes n_1}\otimes \HH^{\otimes n_2}$ if it can be written as the tensor product $\ket{\psi}=\ket{\psi_1}\ket{\psi_2}$, where $\ket{\psi_1}\in \HH^{\otimes n_1}$,  $\ket{\psi_2}\in \HH^{\otimes n_2}$ and $n=n_1+n_2$. A state $\ket{\psi}\in \HH^{\otimes n}$ that is not separable under any partition is called \textit{entangled}. Any unitary transformation $U\in \HH^{2^n\times 2^n}$ acting on $\HH^{\otimes n}$ maps a quantum state $\ket{\psi}\in \HH^{\otimes n}$ to a new quantum state $U\ket{\psi}\in \HH^{\otimes n}$. Some unitary transformations that operate on one or two qubits are called \textit{gates}, and the application of multiple gates is called a \textit{quantum circuit}. 

We adopt the \textbf{quantum query model} of computation based on \textbf{quantum RAM} (QRAM) with quantum registers of size $O(\log n)$ (Word-QRAM).
In this model, we can apply the transformation
\[
    \sum_{i=1}^n\alpha_i\ket{i}\ket{x} \rightarrow \sum_{i=1}^n\alpha_i\ket{i}\ket{x+A[i]}
\]
at unit cost, where $i$, $x$ and $A[i]$, the $i$-th element of array $A$, are stored in registers of up to $O(\log n)$ qubits. When $x=0$, this can be phrased as accessing the elements of $A$ in superposition. Our complexities will then measure the number of queries to the quantum RAM performed by our algorithm. Since we will always perform a constant number of additional operations (e.g. single-character comparison) per queried item, the \textbf{gate complexity} of our algorithm will incur at most in an overhead of logarithmic factors, which could be anyway hidden in the $\tilde{O}$ notation.

The final key ingredient of our algorithms is Grover's search~\cite{Grover96} and its generalization in the framework of amplitude amplification~\cite{Brassard2002}.
\begin{theorem}[\cite{Brassard2002, Grover96}]
\label{thm:amplitude-amplification}
Let $\mathcal{A}$ be a quantum algorithm with no measurement, such that $\mathcal{A}|0\rangle = \sqrt{1-a}|\psi_0\rangle + \sqrt{a}|\psi_1\rangle$, where $|\psi_1\rangle$ denotes a superposition of target states, $|\psi_0\rangle$ is a superposition of the non-target states and $a$ is the success probability ($0 < a \le 1$). 
There exists a quantum algorithm that finds a target state with probability at least $\max(1-a, a)$ using $O(1/\sqrt{a})$ applications of $\mathcal{A}$ and $\mathcal{A}^{-1}$.
\end{theorem}
If $\mathcal{A}$ is a classical function, this corresponds to standard Grover's search. Since $\mathcal{A}$ can also be a Grover's search itself, this result allows us to nest Grover's searches one into the other to achieve better speed-ups. In our context, the success probability $a$ is the ratio between the number of solutions over the number of possibilities. For example, if we compare two strings both of length $n$ and we look for a single-character mismatch, then $a=\frac mn$, where $m$ is the number of single-character mismatches. Then, the complexity of finding one mismatch is $O\left(\frac{1}{\sqrt{a}}\right)=O\left(\sqrt{\frac nm}\right)$. This goes to $O\left(\sqrt{n}\right)$ in the worst case, that is when $m=1$. When the error probability of the oracle function is two-sided, that is both false positives and false negatives are possible, this framework can still be used with at most a $O(\log n)$-factor overhead~\cite{HMDeW03}.

\section{Classical Parallel Algorithm}
\label{sec:main-classical-algo}
We first give a high-level idea of a classical parallel algorithm finding a match for a pattern string $P$ in a GD string $T$. This helps us set up the intuition for the quantum algorithm. The strategy is to use $m=|P|$ threads $t_1,\ldots,t_m$ to compute information about the prefixes of pattern $P$ while we scan GD string $T$, using the forward and backward tries of the segments of $T$. The purpose of this algorithm is not to be efficient, rather provide a framework to better interpret the quantum algorithm. Moreover, for the sake of exposition, we assume that $k_i<m$ for every $1\leq i\leq n$, both here and in the quantum algorithm. This implies that $P$ does not occur as a substring of any string in $T$. We will drop this assumption in Section~\ref{sec:dropping-assumptions} by designing a custom quantum algorithm for finding such occurrences, and showing that the correctness of our main algorithm is not affected. These two cases could be covered in one single algorithm, but we find that both the explanation and the analysis are much cleaner if they are accounted for separately.

\subsection*{Algorithm idea}
The parallel classical algorithm proceeds as follows. After instantiating the $m$ threads, we start scanning $T$ from left to right, segment by segment, in a for-loop of $n$ iterations. At first, let us focus on what thread $t_1$ does. Starting from the first segment $T[1]$ at iteration $1$, thread $t_1$ checks whether $P[1..k_1]\in T[1]$ using the forward trie of $T[1]$. At iteration $2$, $t_1$ checks whether $P[k_1+1..k_1+k_2]\in T[2]$, assuming $k_1+k_2<m$. At some later iteration $i$, $t_1$ will reach the end of $P$, and it will try to match a suffix of $P$ as a prefix of a string in $T[i]$, ending at a certain column $c$. To start checking a match from column $c+1$, $t_1$ tries to match a prefix of $P$ as a suffix of a string in $T[i]$ using the backward trie of $T[i]$. Then, $t_1$ proceeds as before and continues to try to match $P$ in $T$ in this manner until reaching the end of $T$. This means that, after scanning $T$, $t_1$ can tell whether there is a match for $P$ in $T$ starting from some column $1+m\cdot r$, for some integer $r$ such that $1+m\cdot r<W$. Any other generic thread $t_h$ does the same as $t_1$, but it starts with a shift of $h$ columns. In other words, a generic thread $t_h$ can tell whether there is a match for $P$ in $T$ starting from some column $h+m\cdot r$. As depicted in Figure~\ref{fig:shifted_threads}, it is now straightforward to see that, if there is a match for $P$ in $T$, there will be a thread able to find it, because $h$ ranges from $1$ to $m$, thus for every column there's always a thread trying to start a match from that column.
\begin{figure}
    \centering
    \includegraphics[width=1\linewidth]{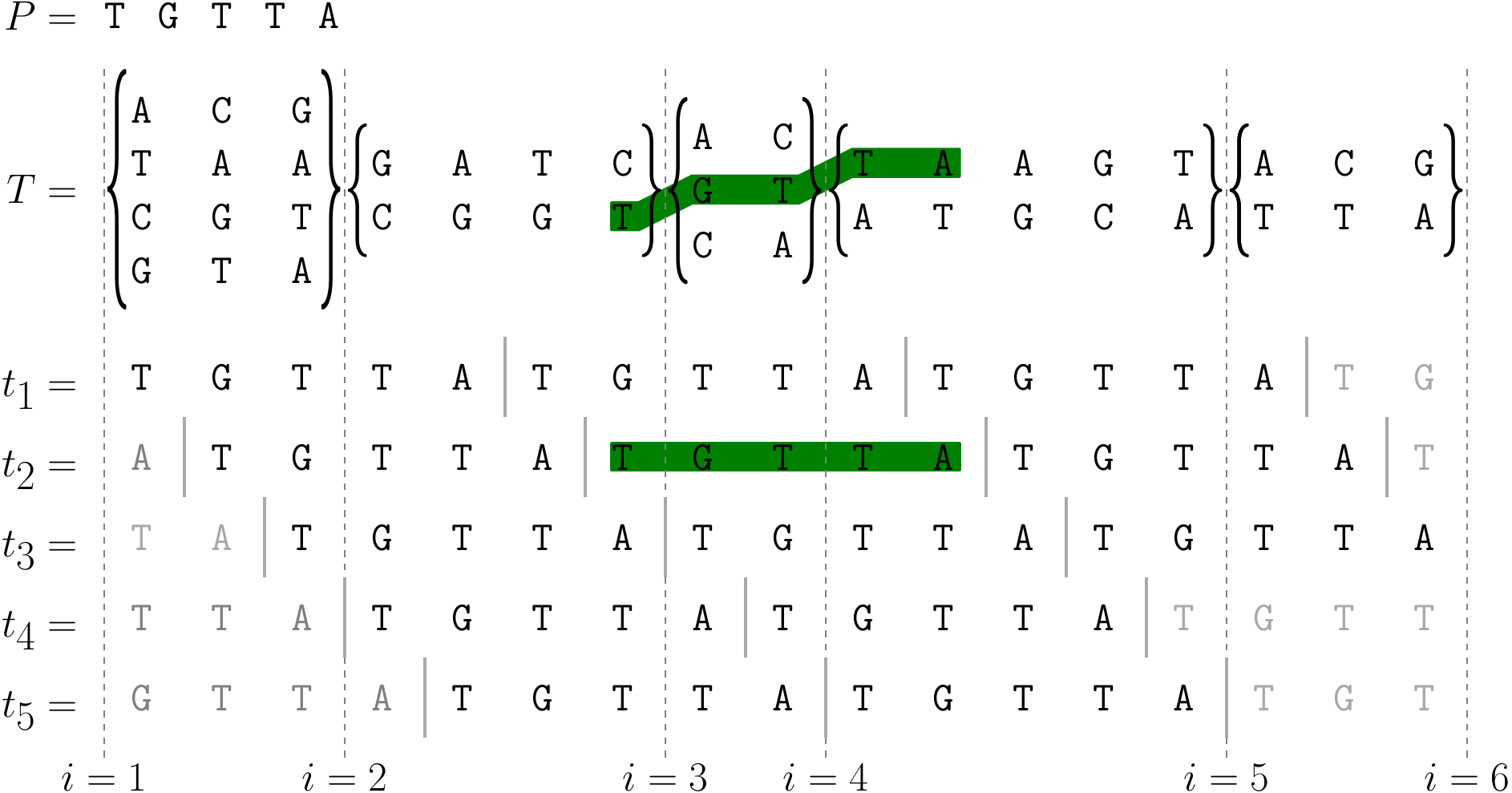}
    \caption{Abstract representation of different threads trying to match pattern $P$ in GD string $T$ starting from different position. Each dash symbol represents a single character, thus $|P|$ has $m=5$ characters and $T$ has $N=52$. Each thread $t_h$ tries to match $P$ column by column, with a shift of $h-1$ positions w.r.t. thread $t_1$, namely $t_1$ is shifted by $0$ positions and $t_5$ is shifted by $2$ positions. The characters highlighted in green show that thread $t_2$ finds a match at position $h + r\cdot m=2+1\cdot 5=7$. The grayed-out characters represents comparisons that will be tested but that cannot become full matches. Variable $i$ counts the iterations of the main for-loop.}
    \label{fig:shifted_threads}
\end{figure}

\section{Quantum Algorithm}
\label{sec:main-quantum-algo}
The quantum algorithm simulates the parallel one by replacing the $m$ threads with a superposition of $m$ states, over which we will run a Grover's search. We recall that here we assume that $k_i<m$ for every $1\leq i\leq n$, and we show how to drop this assumption in Section~\ref{sec:dropping-assumptions}. The oracle function of the Grover's search tells whether a specific shift of the pattern allows to find a match. We first present the quantum algorithm with a classical implementation of the oracle function, that is each quantum thread of the oracle function performs the same computation of the corresponding classical thread of the parallel algorithm. We later improve this result by replacing the use of tries with nested Grover's searches.

\subsection{Simplified quantum algorithm}
\label{sec:simplified-quantum-algo}
The parallel algorithm can easily be turned into a quantum algorithm. We start by preparing superposition $\frac{1}{\sqrt{m}}\sum_{h'=0}^{m-1}\ket{h'}\ket{0\cdots 0}$, where the first register is the one over which we run the Grover's search, while the rest is all the auxiliary qubits we will need to implement the oracle function. The oracle function does exactly the same as the classical threads: given the values $h'$ that identify the shift of the pattern, it scans GD string $T$ from left to right in a for loop of $n$ iterations. For ease of exposition, we use values $h=h'+1$ as thread identifiers. At each iteration $i$, substring $P[j_h..j_h+k_i-1]$ could match in segment $T[i]$, and whether this is the case can be checked by using the forward trie of $T[i]$. Index $j_h=1+((\sum_{j=1}^{i-1} k_j)-h+1\text{ mod } m)$ depends on the specific thread, and the details of this calculation are given in Lemma~\ref{lemma:invariant}, Section~\ref{sec:analysis}. It can be the case that $j_h+k_i-1$ exceeds the length of the pattern, that is we have to stop early our visit of the forward trie because we have characters to match only a suffix of $P$ in $T[i]$. In this case, we also use the backward trie of $T[i]$ to check whether a prefix of $P$ can start a match in $T[i]$. The oracle function implemented in this way scans $T$ from left to right taking time $O(k_i)$ for each individual segment thanks to the tries, resulting in an overall time complexity of $O(W)$. Since the size of our superposition is $O(m)$, that is we have $m$ quantum threads, using Theorem~\ref{thm:amplitude-amplification} our algorithm has a total time complexity of $O(W\sqrt{m})$ and constant success probability. Adding the preprocessing time $O(N)$ needed to compute the tries, we obtain the following intermediate result.
\begin{theorem}
    \label{thm:simple-quantum-algo}
    There exists a quantum algorithm that solves \smgd on a pattern string $P$ of length $m$ and a generalized degenerate string $T$ of $n$ segments and width $W$ in $\tilde{O}\left(N+W\sqrt{m}\right)$ time, with success probability $p>2/3$.
\end{theorem}

\subsection{Final Quantum Algorithm}
Our final quantum algorithm combines the strategy described above and nested Grover's searches to achieve the best performance. First, let us introduce a few useful definitions. Looking at the parallel and quantum algorithms proposed so far we see that we need to compute the following Boolean functions:
\begin{itemize}
    \item $\shiftm(h)=1$ if and only if $P$ has a match starting from some column $h+r\cdot m$ in $T$, for some integer $r$;
    \item $\fullm(s,i)=1$ if and only if string $s$ matches a string in segment $T[i]$, that is $s \in T[i]$.
    \item $\forwm(s,i)=1$ if and only if string $s$ matches the prefix of a string in segment $T[i]$.
    \item $\backm(s,i)=1$ if and only if string $s$ matches the suffix of a string in segment $T[i]$.
\end{itemize}
In the parallel implementation of the algorithm, the threads compute function $\shiftm$ by scanning $T$ from left to right, and so does the quantum algorithm, where $\shiftm$ is the oracle function for the Grover's search. To compute $\shiftm$, we need to compute $\fullm$, $\forwm$ and $\backm$ as subroutines for each segment $T[i]$, and so far we used tries to achieve this task.

We now show how the quantum algorithm can do without tries by the means of Grover's search. We present the procedure for computing function $\fullm(s,i)$. The same logic applies for computing $\forwm(s,i)$ and $\backm(s,i)$ by simply ignoring the unnecessary part of a segment. We start by first recalling that a simple way to check whether two strings $A$ and $B$ of the same length are equal~\cite{RV03} is to run a Grover's search over the positions of $A$ and $B$, with the oracle function marking single-character mismatches. 
We can use this as a subroutine in computing $\fullm(s,i)$. Given string $s$ and segment $T[i]$, we run a Grover's search over all strings in $T[i]$, where the oracle function checks whether two strings are the same. Using the amplitude amplification framework, we can then implement the oracle function with a Grover's search that finds single-character mismatches between $s$ and a string in $T[i]$. To achieve our final quantum algorithm, consider the simple quantum algorithm described in Section~\ref{sec:simplified-quantum-algo}. We replace the use of tries for computing $\fullm$, $\forwm$ and $\backm$ with $O(\log n)$ executions per segment of the nested Grover's searches just described. The reason of the $O(\log n)$ number of executions is to guarantee to bound the failure probability by a constant, as we need to run $n$ Grover's searches, one per segment, to compute function $\shiftm$. The details of the analysis are given in Section~\ref{sec:analysis} in the proof of Lemma~\ref{lemma:f_1-computation}.

\subsection{Details of the Main For Loop}
We specify a few more implementation details, that will be useful later for proving the correctness of the algorithm. Consider the $m$ quantum threads of the outermost Grover's search and the $n$ iteration performed to scan GD string $T$ while computing the oracle function $\shiftm$. We define values $j_{i,h}$, $a_{i,h}$ and $m_{i,h}$ that are integers stored by every quantum thread $h$, using different registers for every iteration $i$, as follows.
\begin{itemize}
    \item $j_{i,h} \in [1,m]$ identifies which prefix of $P$ is managed by quantum thread $h$ at iteration $i$. \item $a_{i,h}=1$ if $P[1..j_{i,h}+k_i-1]$ matches a suffix of $T[1]\cdots T[i]$ after iteration $i$, $a_{i,h}=0$ otherwise.
    \item $m_{i,h} = 1$ if at least one full match for $P$ was found in GD string $T[1]T[2]\cdots T[i]$ after iteration $i$ or during some previous iteration, $m_{i,h}=0$ otherwise;
\end{itemize}
Our algorithm will compute these values as specified in Subroutine~\ref{algo:variable-updates}. Intuitively, $a_{i,h}$ tracks whether shift $h$ of the pattern currently has an active prefix matching a suffix of the GD string scanned so far. We flip $a_{i,h}$ from $0$ to $1$ in those iterations in which we match a prefix of $P$ as a suffix of the current segment $T[i]$. In any later iteration $i'$, as long as we can extend the match to the current segment, we keep $a_{i',h}$ set to $1$, otherwise we set it to $0$. Value $m_{i,h}$ tracks whether we found a match at any previous iteration. If yes, $m_{i,h}$ is set to $1$ for the rest of the computation of $\shiftm$, and the corresponding thread $h$ will be a marked element in the Grover's search.
 
\RestyleAlgo{boxed}
\begin{algorithm}
    \SetAlgorithmName{Subroutine}{Subroutine}{List of subroutines}
    \eIf{$j_{i,h}+k_i-1 < m$}{
        $a_{i+1,h}\gets a_{i,h} \land \fullm(P[j_{i,h}..j_{i,h}+k_i-1],i)$\;
        $m_{i+1,h} \gets m_{i,h}$
    }{\tcp{$j_{i,h}+k_i-1 \geq m$}
    $a_{i+1,h}\leftarrow \backm(P[1..j_{i,h}+k_i-1-m],i)$\;
    $m_{i+1,h} \leftarrow m_{i,h} \lor (\forwm(P[j_{i,h}..m],i))\land a_{i,h})$\;
    }
    $j_{i+1,h}\gets 1 + (j_{i,h}+k_i-1 \mod m)$
\caption{Variable updates for one iteration of the computation of $\shiftm$. In the update of $j_{i+1,h}$, we subtract $1$ within the modulo operation and add it back outside to keep $j_{i+1,h}$ ranging from $1$ to $m$ as a function of $h$.  For each iteration $i$, we use new registers initialized to $\ket{0}$ to store the new values. }
\label{algo:variable-updates}
\end{algorithm}

\subsection{Reporting Occurrences}
\label{subsec:reporting-occurrences}
In order to report the occurrences, we use the output of our main algorithm. After running the outermost Grover's search and measuring, we obtain a value $h\in[1,m]$ identifying the quantum thread that was successful in finding an occurrence. Thus, we can just simulate the execution of that thread, that is we can start from column $h$ and test every character of pattern $P$ against the corresponding column of the GD string, restarting from the first character of $P$ every $m$ columns. Notice that, even if we use quantum subroutines for computing $\fullm$, $\forwm$ and $\backm$, their computation is performed at the level of a single segment, thus we can measure after processing each segment. Thus, after every iteration $i$ of our main for loop, we can check whether $m_{i,h}=1$. If yes, let $c_{i}$ be the index of the first column of segment $T[i]$, which we can retrieve in constant time by definition of GD string. Since we know that the match starts in $T[i']$, for some $i'<i$, and ends in $T[i]$, we compute the starting position of the match as
\[
c_{\operatorname{start}} = c_i - (c_i-h+1 \mod m)
\]
and the ending position as $c_{\operatorname{end}} = c_{\operatorname{start}}+m-1$. We recall that $h\in[1,m]$ is the index of the thread, and thus $h-1$ is the amount of positions by which the pattern is shifted for thread $h$.

\section{Analysis}
\label{sec:analysis}
We start by proving some intermediate lemmas, useful to show the correctness of our algorithm. Then, we combine these with an analysis of the complexity to prove our main result.
\begin{lemma}
    \label{lemma:invariant}
    The for loop of our quantum algorithm computing oracle function $\shiftm$ maintains the following invariant right after every iteration $i$:
    
    \textbf{Invariant}: for every $h\in[1,m]$ it holds that:
    \begin{itemize}
        \item $j_{i,h}=1+\left(\left(\sum_{t=1}^{i}k_t\right) - h+1 \mod m\right)$;
        \item $j_{i,h+1}=j_{i,h}+1$, if $h\in[1,m-1]$;
        \item $a_{i,h}=1$ if and only if $P[1..j_{i,h}-1]$ matches a suffix of $T[1]\cdots T[i]$;
        \item $m_{i,h}=1$ if and only if $P$ has a match in $T[1]\cdots T[i]$.
    \end{itemize}
\end{lemma}
\begin{proof}
    To prove the correctness of the updates, we proceed by induction on the number of performed iterations $i$.
    
    \textbf{Base case}, $i=0$. After the system is initialized with all threads active, and before performing the first iteration, we have $j_{0,h}=1+(-h+1 \mod m)$, $a_{0,h}=0$, $m_{0,h}=0$ for every $h \in [1,m]$. The invariant trivially holds for $a_{0,h}$ and $m_{0,h}$, as no text has been processed and no partial or full matches exist beyond the empty prefix. For $j_{0,h}$, we have $\sum_{t=1}^{0}k_t=0$ since we processed no segments. We then recall that $h$ is the thread identifier, thus $j_{0,h}$ is shifted by $h-1$ positions. The first thread starts from the first position $j_{0,1}=1+(0\mod m)=1$; the second has a shift of one position, thus starting at $j_{0,2}=1+(-1\mod m)=m$; the third has a shift of two positions, starting from $j_{0,3}=1+(-2\mod m)=m-1$, and so on, with the $h$-th starting at
    \[j_{0,h}=1+(-(h-1) \mod m)=1+(-h+1 \mod m).\]
    
    \textbf{Inductive case}, $i>1$. We assume the invariant holds after iteration $i$, and we prove that it still holds after iteration $i+1$, by showing that the values for every thread $h$ are updated correctly. The new values of $j_{i,h}$, $a_{i,h}$ and $m_{i,h}$ are computed according to Subroutine~\ref{algo:variable-updates}. Let us first handle the indexes. We always increment $j_{i,h}$ by $k_i$ at the end of the iteration, and since by inductive hypothesis $j_{i,h}=1+\left(\left(\sum_{t=1}^{i}k_t\right) - h+1\right) \mod m$, we have
    \begin{align*}
        j_{i+1,h}   =& 1+\left(j_{i,h} + k_i - 1 \mod m\right)\\
                    =& 1+\left(1+\left(\sum_{t=1}^{i}k_t \right) - h+1 + k_i-1 \mod m\right)\\
                    =& 1+\left(\left(\sum_{t=1}^{i+1}k_t\right) - h+1 \mod m\right)
    \end{align*}
    Adding the same value $k_i$ for every $h$ also maintains the property that $j_{i+1,h+1}=j_{i+1,h}+1$. For $a_{i,h}$ and $m_{i,h}$, we have two cases: (i) $j_{i+1,h}+k_{i+1}-1 < m$, (ii) $j_{i+1,h}+k_{i+1}-1 \geq m$.
    
    In case (i), if $a_{i,h}=0$ then by inductive hypothesis prefix $P[1..j_{i,h}-1]$ is not matching, and we correctly set $a_{i+1,h}=0$. Otherwise, $a_{i,h}=1$ and by inductive hypothesis the prefix $P[1..j_{i,h}-1]$ has an active match. If that occurrence extends into $T[i+1]$, there must exists a string $s \in T[i+1]$ matching $P[j_{i,h}..j_{i,h}+k_{i+1}-1]$. Whether such a string exists or not is determined by the result of $\fullm(P[j_{i,h}..j_{i,h}+k_{i+1}-1],i+1)$, and we are correctly setting $a_{i+1,h}$ to this value. For what concerns $m_{i+1,h}$, since no new match can be detected in this case, it is correct to carry over the previous value $m_{i,h}$.
    
    In case (ii), an occurrence of $P$ might end in $T[i+1]$, a new one could start, and the record of an old one could carry over. Whether a string $s \in T[i+1]$ has a prefix matching $P[j_{i,h}..m]$ or not is found by computing $\forwm(P[j_{i,h}..m],i+1)$. Then, $m_{i+1,h}$ is set to $1$ if both $\forwm(P[j_{i,h}..m],i+1)=1$ and $a_{i,h}=1$, that is suffix $P[j_{i,h}..m]$ matches in the current segment, and by inductive hypothesis there is an active match for prefix $P[1..j_{i,h}-1]$. Otherwise, no full match can be found, and $m_{i+1,h}$ remains $0$. If it was the case that $m_{i,h}=1$, by inductive hypothesis we already found a match at a previous iteration, thus we correctly set $m_{i+1,h}=1$. Finally, $a_{i+1,h}$ must be $1$ when prefix $P[1..j_{i+1,h}-1]$ matches a suffix of a string in $T[i+1]$, which is done with $\backm(P([1..j_{i+1,h}-1],i+1)$.
\end{proof}

\begin{lemma}
\label{lemma:f_1-computation}
    Let GD string $T$ have $n$ segments $T[i]$ each of size $k_i$. The for-loop of our quantum algorithm correctly computes oracle function\\ $\shiftm(h)$ for each quantum thread $h$ in time $\tilde{O}(\sum_{i=1}^n\sqrt{T[i]\cdot k_i})$ with constant probability of success $p > 2/3$.
\end{lemma}
\begin{proof}
    Using Lemma~\ref{lemma:invariant} we know that, after iteration $n$, we have $m_{n,h}=1$ if and only if thread $h$ found at least one match for $P$ in $T[1]\cdots T[n]=T$, which guarantees the correctness of the computation of $\shiftm(h)$ for every $h$. Notice that the innermost Grover's search gives a negative answer when it fails to find single-character mismatches, but this is used as a positive answer for the oracle of the Grover's search performed at the level of a segment, meaning we can have two-sided errors. Consequently, the correctness of the entire algorithm follows from Lemma~\ref{lemma:invariant}, Theorem~\ref{thm:amplitude-amplification} for nesting Grover's searches, its generalizations to two-sided errors scenarios~\cite{HMDeW03}, and the following observation. At any iteration $i$ of the for-loop, for every column in $T[i]$ there exists a thread $h$ that tries to start a match from that column. To see this, recall that $j_{i,h}$ is a position in $P$ and that a thread $h$ tries to match prefix $P[1..j_{i,h}+k_i-1]$ as a suffix of a string in $T[i]$ when $j_{i,h}+k_i-1 \geq m$. Thus, $j_{i,h}+k_i-1$ assumes all values between $1$ and $k_i$ as a function of $h$, which means there is a thread trying to start a match for every column.

    The success probability of $p>2/3$ is obtained by bounding the failure probability $q=1-p$ of computing $\shiftm$ by $1/3$. To see this, let us consider the failure probability $q_i$ of the Grover's search for a single segment $T[i]$. Since here we are nesting two Grover's searches, we get the failure probability bounded by $1/3$. We run this procedure $\log_3(c n)$ times, thus bounding the error probability as $q_i<1/3^{\log_3(c n)}$, for some constant $c$. Since we fail as long as any of the $n$ segment-level Grover's searches  fails, we can use the union bound to obtain that $q \leq \sum_{i=1}^n q_i \leq n/3^{\log_3(c n)}$. By choosing $c=4$, we get $q \leq 1/4 < 1/3$, implying $p>2/3$.

    For what concerns the complexity, at iteration $i$, functions $\fullm$, $\forwm$ and $\backm$ are computed using two nested Grover's searches. The inner Grover's search checks for single-character mismatches over strings of length no more than $k_i$, thus its complexity is $O\left(\sqrt{k_i}\right)$. The outer Grover's search runs over $|T[i]|$ many strings, using the inner search as oracle function, yielding a complexity of $O\left(\sqrt{|T[i]|\cdot k_i}\right)$ for processing one segment. For each segment, we run the outer Grover's search $O(\log n)$ times. Since we process $n$ segments sequentially, we have a total complexity of $\tilde{O}(\sum_{i=1}^n\sqrt{T[i]\cdot k_i})$.
\end{proof}

In order to give the final complexity analysis, we need to show how to bound the complexity of the oracle function $\shiftm$.
\begin{lemma}
\label{lemma:bound}
Let $N = \sum_{i=1}^n |T[i]| k_i$ be the total length of all strings constituting the GD string $T$. The following inequality holds:
$$ \sum_{i=1}^n \sqrt{|T[i]| k_i} \le \sqrt{nN} $$
\end{lemma}

\begin{proof}
We apply the Cauchy-Schwarz inequality, which states that for any sequences of real numbers $(a_1, \dots, a_n)$ and $(b_1, \dots, b_n)$,
$$ \left( \sum_{i=1}^n a_i b_i \right)^2 \le \left( \sum_{i=1}^n a_i^2 \right) \left( \sum_{i=1}^n b_i^2 \right) $$
\begin{itemize}
    \item Let $a_i = 1$ for all $i = 1, \dots, n$.
    \item Let $b_i = \sqrt{|T[i]| k_i}$.
\end{itemize}

Substituting these into the inequality we have,
$$ \left( \sum_{i=1}^n 1 \cdot \sqrt{|T[i]| k_i} \right)^2 \le \left( \sum_{i=1}^n 1^2 \right) \left( \sum_{i=1}^n \left(\sqrt{|T[i]| k_i}\right)^2 \right) $$
where 
$$\sum_{i=1}^n 1^2 = \sum_{i=1}^n 1 = n \quad\text{and}\quad \sum_{i=1}^n \left(\sqrt{|T[i]| k_i}\right)^2 = \sum_{i=1}^n |T[i]| k_i.$$

By definition, $\sum_{i=1}^n |T[i]| k_i$ is exactly $N$, the total size of the GD string.

Therefore,
$$ \left( \sum_{i=1}^n \sqrt{|T[i]| k_i} \right)^2 \le n \cdot N $$

Taking the square root of both sides, we obtain,
$$ \sum_{i=1}^n \sqrt{|T[i]| k_i} \le \sqrt{nN} $$
\end{proof}

Our main result Theorem~\ref{thm:main-quantum-algo} states that \smgd on a pattern string $P$ of length $m$ and a GD string of $n$ segments and size $N$ can be solved in time $\tilde{O}(\sqrt{mnN})$ by a quantum algorithm with high probability. We are now ready to formally prove this result.
\begin{proof}[Proof of Theorem~\ref{thm:main-quantum-algo}]
    From Lemma \ref{lemma:f_1-computation}, if an occurrence of $P$ exists in $T$, there exists at least one thread $h$ such that $m_{n,h}=1$. The algorithm applies a Grover's search to compute $\shiftm$ to the equally balanced superposition $\sum_{h=0}^m\ket{h}\ket{0}\otimes\cdots\otimes\ket{0}$, and from the properties of Grover's search~\cite{Grover96} follows that measuring the system after at most  $O(\sqrt{m})$ iterations yields a state $\ket{h}$ such that $\shiftm(h)=1$ with probability $p>2/3$, if at least one match exists. If no match exists, the oracle never marks any state and the measurement will return a random thread index. We can verify whether $\shiftm(h)=1$ or $\shiftm(h)=0$ by running one more evaluation of the oracle function. In doing so, we can also report one occurrence, by performing the steps described in Section~\ref{subsec:reporting-occurrences} as we recompute the oracle function.

    As mentioned above, the computation of this Grover's search requires $O(\sqrt{m})$ queries to oracle function $\shiftm$. Lemma~\ref{lemma:f_1-computation} tells us that computing the oracle function takes $\tilde{O}\left(\sum_{i=1}^n\sqrt{T[i]\cdot k_i}\right)$, and using Lemma~\ref{lemma:bound} this can be bounded by $\tilde{O}\left(\sqrt{nN}\right)$. This leads to a final complexity of $\tilde{O}\left(\sqrt{mnN}\right)$, with an additional $\tilde{O}\left(\sqrt{nN}\right)$ needed to recompute the oracle function and report one occurrence.
\end{proof}

For what concerns space complexity, we need $O(\log m)$ qubits for the register representing thread IDs, and $O(\log N)$ for the registers introduced at each iteration, as values $k_i$ could be order of $O(N)$ in the worst case. Since we need to introduce $n$ of these registers, one per iteration, the total space complexity is $O(\log m+n\log N)$ qubits.

\section{Dropping Assumptions}
\label{sec:dropping-assumptions}
So far, we assumed that $k_i<m$ for every $1\leq i\leq n$. Let us now drop this assumption. We can still find a match for $P$ in $T$ without losing on the complexity. To this end, let us first design a quantum algorithm that checks whether $P$ has a match as a substring of a string in $T$. We achieve this with two nested Grover's searches, the outer one ranging over all $N$ characters in $T$, and the inner one ranging over all $m$ pattern positions in $P$. The inner Grover's search uses an oracle function that detects two things: single characters mismatches, and whether the character that we are checking is the last one in a segment string when we are not on the last position of the pattern. The outer Grover's search uses the inner one as oracle function. Thus, this takes $O\left(\sqrt{mN}\right)$ in total, which is less than the overall complexity $\tilde{O}\left(\sqrt{mnN}\right)$.
 
If the above preprocessing procedure found no matches, we can continue with our main algorithm assuming that $P$ does not match as a substring of a string in $T$. However, it can still be the case that $k_i>m$ for some $i$, requiring us to make a small modification to the algorithm. In the main for loop, whenever we encounter a segment $T[i]$ such that $k_i>m$, we do not compute $\fullm$, but only $\forwm$ and $\backm$. Now it is no longer true what we showed in the proof of Lemma~\ref{lemma:f_1-computation}, that is not all columns will have a quantum thread that tries to start a match. However, it is easy to see that the columns that we do not cover are the ones already checked by the preprocessing procedure explained above.

\section{Rephrasing Previous Results}
\label{sec:rephrasing}
The quantum algorithm of Equi et al.~\cite{EGM23} simulates a bit-parallel algorithm run over a level-DAG, where a bit-vector is computed for every node. The quantum superposition simulates the bit-vector, and the crucial operation is the bitwise shift by one position, which is simulated by incrementing the identifier of each quantum thread by $1$. Although this is just an abstraction and no interference happens among the quantum threads, seeing the problem from this angle makes it hard to generalize the approach to other structures like GD strings, and it is not clear how to retrieve an occurrence. In particular, it is hard to reason about what computation a single quantum thread does, as in the bit-parallel abstraction the classical algorithm heavily relies on the shift operation, which brings a new bit of information from the previous bit-vector position to the current one. 

If we cast this algorithm in our framework, it is clear that each quantum thread is handling a different shift of the pattern. At level $l$, some nodes in the level are marked as having an active prefix match. Quantum thread $h$ scans all these nodes and checks which ones match the current pattern character. If a node has a matching character, then all the out-neighbors of that node are marked as having an active prefix match. Then, the computation continues to the next level. In the end, the result of the Grover's search is the index of a thread that found a match. To locate an occurrence, simply run that single thread again classically, checking whether a match is active every time character $P[m]$ is tested. We defer the formal details of this discussion to the extended version of this work.

\bibliography{biblio}

@inproceedings{Asconeetal24,
  author       = {Rocco Ascone and
                  Giulia Bernardini and
                  Alessio Conte and
                  Massimo Equi and
                  Est{\'{e}}ban Gabory and
                  Roberto Grossi and
                  Nadia Pisanti},
  editor       = {Solon P. Pissis and
                  Wing{-}Kin Sung},
  title        = {A Unifying Taxonomy of Pattern Matching in Degenerate Strings and
                  Founder Graphs},
  booktitle    = {24th International Workshop on Algorithms in Bioinformatics, {WABI}
                  2024, Royal Holloway, London, United Kingdom, September 2-4, 2024},
  series       = {LIPIcs},
  volume       = {312},
  pages        = {14:1--14:21},
  publisher    = {Schloss Dagstuhl - Leibniz-Zentrum f{\"{u}}r Informatik},
  year         = {2024},
  url          = {https://doi.org/10.4230/LIPIcs.WABI.2024.14},
  doi          = {10.4230/LIPICS.WABI.2024.14},
  bibsource    = {dblp computer science bibliography, https://dblp.org}
}

@inproceedings{EGM23,
  author       = {Massimo Equi and
                  Arianne Meijer{-}van de Griend and
                  Veli M{\"{a}}kinen},
  editor       = {Laurent Bulteau and
                  Zsuzsanna Lipt{\'{a}}k},
  title        = {From Bit-Parallelism to Quantum String Matching for Labelled Graphs},
  booktitle    = {34th Annual Symposium on Combinatorial Pattern Matching, {CPM} 2023,
                  Marne-la-Vall{\'{e}}e, France, June 26-28, 2023},
  series       = {LIPIcs},
  volume       = {259},
  pages        = {9:1--9:20},
  publisher    = {Schloss Dagstuhl - Leibniz-Zentrum f{\"{u}}r Informatik},
  year         = {2023},
  url          = {https://doi.org/10.4230/LIPIcs.CPM.2023.9},
  doi          = {10.4230/LIPICS.CPM.2023.9},
  bibsource    = {dblp computer science bibliography, https://dblp.org}
}

@inproceedings{QAC25,
author = {Khadiev, Kamil and Serov, Danil},
title = {Quantum Algorithm for the Multiple String Matching Problem},
year = {2025},
isbn = {978-3-031-82696-2},
publisher = {Springer-Verlag},
address = {Berlin, Heidelberg},
url = {https://doi.org/10.1007/978-3-031-82697-9_5},
doi = {10.1007/978-3-031-82697-9_5},
booktitle = {SOFSEM 2025: Theory and Practice of Computer Science: 50th International Conference on Current Trends in Theory and Practice of Computer Science, SOFSEM 2025, Bratislava, Slovak Republic, January 20–23, 2025, Proceedings, Part II},
pages = {58–69},
numpages = {12},
location = {Bratislava, Slovakia}
}

@article{KMP77,
  author    = {Donald E. Knuth and
               James H. Morris Jr. and
               Vaughan R. Pratt},
  title     = {Fast Pattern Matching in Strings},
  journal   = {{SIAM} J. Comput.},
  volume    = {6},
  number    = {2},
  pages     = {323--350},
  year      = {1977},
  url       = {https://doi.org/10.1137/0206024},
  doi       = {10.1137/0206024},
  bibsource = {dblp computer science bibliography, https://dblp.org}
}

@book{NC10QCQI,
    place={Cambridge},
    title={Quantum Computation and Quantum Information: 10th Anniversary Edition},
    DOI={10.1017/CBO9780511976667},
    publisher={Cambridge University Press},
    author={Nielsen, Michael A. and Chuang, Isaac L.},
    year={2010}
}

@article{RV03,
    title = {String matching in {$O(\sqrt{n}+\sqrt{m})$} quantum time},
    author={Ramesh, Hariharan and Vinay, V},
    journal = {Journal of Discrete Algorithms},
    volume = {1},
    number = {1},
    pages = {103-110},
    year = {2003},
    note = {Combinatorial Algorithms},
    issn = {1570-8667},
    doi = {10.1016/S1570-8667(03)00010-8}
}

@inproceedings{DGT22,
  author    = {Parisa Darbari and
               Daniel Gibney and
               Sharma V. Thankachan},
  title     = {Quantum Time Complexity and Algorithms for Pattern Matching on Labeled
               Graphs},
  booktitle = {String Processing and Information Retrieval - 29th International Symposium,
               {SPIRE} 2022, Concepci{\'{o}}n, Chile, November 8-10, 2022, Proceedings},
  series    = {Lecture Notes in Computer Science},
  volume    = {13617},
  pages     = {303--314},
  publisher = {Springer},
  year      = {2022},
  url       = {https://doi.org/10.1007/978-3-031-20643-6\_22},
  doi       = {10.1007/978-3-031-20643-6\_22},
  bibsource = {dblp computer science bibliography, https://dblp.org}
}

@inproceedings{Grover96,
  author    = {Lov K. Grover},
  title     = {A Fast Quantum Mechanical Algorithm for Database Search},
  booktitle = {Proceedings of the Twenty-Eighth Annual {ACM} Symposium on the Theory
               of Computing, Philadelphia, Pennsylvania, USA, May 22-24, 1996},
  pages     = {212--219},
  publisher = {{ACM}},
  year      = {1996},
  url       = {https://doi.org/10.1145/237814.237866},
  doi       = {10.1145/237814.237866},
  bibsource = {dblp computer science bibliography, https://dblp.org}
}

@incollection{Brassard2002,
  author    = {Brassard, Gilles and H{\o}yer, Peter and Mosca, Michele and Tapp, Alain},
  title     = {Quantum amplitude amplification and estimation},
  booktitle = {Quantum Computation and Information},
  series    = {Contemporary Mathematics},
  volume    = {305},
  pages     = {53--74},
  year      = {2002},
  publisher = {American Mathematical Society},
  address   = {Providence, RI},
  isbn      = {0-8218-2140-7},
  doi       = {10.1090/conm/305/05215}
}

@article{LeGallSeddighin23,
  author       = {Fran{\c{c}}ois Le Gall and
                  Saeed Seddighin},
  title        = {Quantum Meets Fine-Grained Complexity: Sublinear Time Quantum Algorithms
                  for String Problems},
  journal      = {Algorithmica},
  volume       = {85},
  number       = {5},
  pages        = {1251--1286},
  year         = {2023},
  url          = {https://doi.org/10.1007/s00453-022-01066-z},
  doi          = {10.1007/S00453-022-01066-Z},
  bibsource    = {dblp computer science bibliography, https://dblp.org}
}

@article{JinNogler24,
  author       = {Ce Jin and
                  Jakob Nogler},
  title        = {Quantum Speed-Ups for String Synchronizing Sets, Longest Common Substring,
                  and \emph{k}-mismatch Matching},
  journal      = {{ACM} Trans. Algorithms},
  volume       = {20},
  number       = {4},
  pages        = {32:1--32:36},
  year         = {2024},
  url          = {https://doi.org/10.1145/3672395},
  doi          = {10.1145/3672395},
  bibsource    = {dblp computer science bibliography, https://dblp.org}
}

@article{AkmalJin23,
  author       = {Shyan Akmal and
                  Ce Jin},
  title        = {Near-Optimal Quantum Algorithms for String Problems},
  journal      = {Algorithmica},
  volume       = {85},
  number       = {8},
  pages        = {2260--2317},
  year         = {2023},
  url          = {https://doi.org/10.1007/s00453-022-01092-x},
  doi          = {10.1007/S00453-022-01092-X},
  bibsource    = {dblp computer science bibliography, https://dblp.org}
}

@article{WangYing24,
  author       = {Qisheng Wang and
                  Mingsheng Ying},
  title        = {Quantum Algorithm for Lexicographically Minimal String Rotation},
  journal      = {Theory Comput. Syst.},
  volume       = {68},
  number       = {1},
  pages        = {29--74},
  year         = {2024},
  url          = {https://doi.org/10.1007/s00224-023-10146-8},
  doi          = {10.1007/S00224-023-10146-8},
  bibsource    = {dblp computer science bibliography, https://dblp.org}
}

@article{BoroujeniEGHS21,
  author       = {Mahdi Boroujeni and
                  Soheil Ehsani and
                  Mohammad Ghodsi and
                  MohammadTaghi Hajiaghayi and
                  Saeed Seddighin},
  title        = {Approximating Edit Distance in Truly Subquadratic Time: Quantum and
                  MapReduce},
  journal      = {J. {ACM}},
  volume       = {68},
  number       = {3},
  pages        = {19:1--19:41},
  year         = {2021},
  url          = {https://doi.org/10.1145/3456807},
  doi          = {10.1145/3456807},
  bibsource    = {dblp computer science bibliography, https://dblp.org}
}

@inproceedings{GibneyJKT24,
  author       = {Daniel Gibney and
                  Ce Jin and
                  Tomasz Kociumaka and
                  Sharma V. Thankachan},
  editor       = {David P. Woodruff},
  title        = {Near-Optimal Quantum Algorithms for Bounded Edit Distance and Lempel-Ziv
                  Factorization},
  booktitle    = {Proceedings of the 2024 {ACM-SIAM} Symposium on Discrete Algorithms,
                  {SODA} 2024, Alexandria, VA, USA, January 7-10, 2024},
  pages        = {3302--3332},
  publisher    = {{SIAM}},
  year         = {2024},
  url          = {https://doi.org/10.1137/1.9781611977912.118},
  doi          = {10.1137/1.9781611977912.118},
  bibsource    = {dblp computer science bibliography, https://dblp.org}
}

@inproceedings{GrossiILPPRRVV17,
  author       = {Roberto Grossi and
                  Costas S. Iliopoulos and
                  Chang Liu and
                  Nadia Pisanti and
                  Solon P. Pissis and
                  Ahmad Retha and
                  Giovanna Rosone and
                  Fatima Vayani and
                  Luca Versari},
  title        = {On-Line Pattern Matching on Similar Texts},
  booktitle    = {28th Annual Symposium on Combinatorial Pattern Matching ({CPM})},
  series       = {LIPIcs},
  volume       = {78},
  pages        = {9:1--9:14},
  year         = {2017},
  _url          = {https://doi.org/10.4230/LIPIcs.CPM.2017.9},
  doi          = {10.4230/LIPICS.CPM.2017.9},
  timestamp    = {Thu, 14 Oct 2021 10:30:13 +0200},
  bib_url       = {https://dblp.org/rec/conf/cpm/GrossiILPPRRVV17.bib},
  bibsource    = {dblp computer science bibliography, https://dblp.org}
}

@article{IliopoulosKP21,
  author    = {Costas S. Iliopoulos and
               Ritu Kundu and
               Solon P. Pissis},
  title     = {Efficient pattern matching in elastic-degenerate strings},
  journal   = {Information and Computation},
  volume    = {279},
  pages     = {104616},
  year      = {2021},
  url       = {https://doi.org/10.1016/j.ic.2020.104616},
  doi       = {10.1016/j.ic.2020.104616},
  bibsource = {dblp computer science bibliography, https://dblp.org}
}

@article{RizzoENM24,
  author       = {Nicola Rizzo and
                  Massimo Equi and
                  Tuukka Norri and
                  Veli M{\"{a}}kinen},
  title        = {Elastic founder graphs improved and enhanced},
  journal      = {Theor. Comput. Sci.},
  volume       = {982},
  pages        = {114269},
  year         = {2024},
  url          = {https://doi.org/10.1016/j.tcs.2023.114269},
  doi          = {10.1016/J.TCS.2023.114269},
  bibsource    = {dblp computer science bibliography, https://dblp.org}
}

@inproceedings{AlzamelABGIPPR18,
  author       = {Mai Alzamel and
                  Lorraine A. K. Ayad and
                  Giulia Bernardini and
                  Roberto Grossi and
                  Costas S. Iliopoulos and
                  Nadia Pisanti and
                  Solon P. Pissis and
                  Giovanna Rosone},
  _editor       = {Laxmi Parida and
                  Esko Ukkonen},
  title        = {Degenerate String Comparison and Applications},
  booktitle    = {18th International Workshop on Algorithms in Bioinformatics ({WABI})},
  series       = {LIPIcs},
  volume       = {113},
  pages        = {21:1--21:14},
  year         = {2018},
  _url          = {https://doi.org/10.4230/LIPIcs.WABI.2018.21},
  doi          = {10.4230/LIPICS.WABI.2018.21},
  timestamp    = {Sun, 25 Oct 2020 23:09:19 +0100},
  bib_url       = {https://dblp.org/rec/conf/wabi/AlzamelA0GIPPR18.bib},
  bibsource    = {dblp computer science bibliography, https://dblp.org}
}

@article{AlzamelABGIPPR20,
  author       = {Mai Alzamel and
                  Lorraine A. K. Ayad and
                  Giulia Bernardini and
                  Roberto Grossi and
                  Costas S. Iliopoulos and
                  Nadia Pisanti and
                  Solon P. Pissis and
                  Giovanna Rosone},
  title        = {Comparing Degenerate Strings},
  journal      = {Fundam. Informaticae},
  volume       = {175},
  number       = {1-4},
  pages        = {41--58},
  year         = {2020},
  _url          = {https://doi.org/10.3233/FI-2020-1947},
  doi          = {10.3233/FI-2020-1947},
  timestamp    = {Mon, 03 Jan 2022 21:59:07 +0100},
  bib_url       = {https://dblp.org/rec/journals/fuin/AlzamelABGIPPR20.bib},
  bibsource    = {dblp computer science bibliography, https://dblp.org}
}

@inproceedings{MakinenCENT20,
  author       = {Veli M{\"{a}}kinen and
                  Bastien Cazaux and
                  Massimo Equi and
                  Tuukka Norri and
                  Alexandru I. Tomescu},
 title        = {Linear Time Construction of Indexable Founder Block Graphs},
  booktitle    = {20th International Workshop on Algorithms in Bioinformatics ({WABI})},
  series       = {LIPIcs},
  volume       = {172},
  pages        = {7:1--7:18},
  year         = {2020},
  _url          = {https://doi.org/10.4230/LIPIcs.WABI.2020.7},
  doi          = {10.4230/LIPICS.WABI.2020.7},
  timestamp    = {Thu, 16 Sep 2021 18:08:24 +0200},
  biburl       = {https://dblp.org/rec/conf/wabi/MakinenCENT20.bib},
  bibsource    = {dblp computer science bibliography, https://dblp.org}
}

@inproceedings{EquiNACTM21,
  author       = {Massimo Equi and
                  Tuukka Norri and
                  Jarno Alanko and
                  Bastien Cazaux and
                  Alexandru I. Tomescu and
                  Veli M{\"{a}}kinen},
  editor       = {Hee{-}Kap Ahn and
                  Kunihiko Sadakane},
  title        = {Algorithms and Complexity on Indexing Elastic Founder Graphs},
  booktitle    = {32nd International Symposium on Algorithms and Computation, {ISAAC}
                  2021, Fukuoka, Japan, December 6-8, 2021},
  series       = {LIPIcs},
  volume       = {212},
  pages        = {20:1--20:18},
  publisher    = {Schloss Dagstuhl - Leibniz-Zentrum f{\"{u}}r Informatik},
  year         = {2021},
  url          = {https://doi.org/10.4230/LIPIcs.ISAAC.2021.20},
  doi          = {10.4230/LIPICS.ISAAC.2021.20},
  bibsource    = {dblp computer science bibliography, https://dblp.org}
}

@article{IliopoulosMR08,
  author       = {Costas S. Iliopoulos and
                  Laurent Mouchard and
                  Mohammad Sohel Rahman},
  title        = {A New Approach to Pattern Matching in Degenerate {DNA/RNA} Sequences
                  and Distributed Pattern Matching},
  journal      = {Math. Comput. Sci.},
  volume       = {1},
  number       = {4},
  pages        = {557--569},
  year         = {2008},
  _url          = {https://doi.org/10.1007/s11786-007-0029-z},
  doi          = {10.1007/S11786-007-0029-Z},
  timestamp    = {Wed, 12 Aug 2020 10:36:08 +0200},
  bib_url       = {https://dblp.org/rec/journals/mics/IliopoulosMR08.bib},
  bibsource    = {dblp computer science bibliography, https://dblp.org}
}

@inproceedings{BuhrmanPS21,
  author       = {Harry Buhrman and
                  Subhasree Patro and
                  Florian Speelman},
  editor       = {Markus Bl{\"{a}}ser and
                  Benjamin Monmege},
  title        = {A Framework of Quantum Strong Exponential-Time Hypotheses},
  booktitle    = {38th International Symposium on Theoretical Aspects of Computer Science,
                  {STACS} 2021, Saarbr{\"{u}}cken, Germany (Virtual Conference),
                  March 16-19, 2021},
  series       = {LIPIcs},
  volume       = {187},
  pages        = {19:1--19:19},
  publisher    = {Schloss Dagstuhl - Leibniz-Zentrum f{\"{u}}r Informatik},
  year         = {2021},
  url          = {https://doi.org/10.4230/LIPIcs.STACS.2021.19},
  doi          = {10.4230/LIPICS.STACS.2021.19},
  bibsource    = {dblp computer science bibliography, https://dblp.org}
}

@inproceedings{HMDeW03,
author = {H\o{}yer, Peter and Mosca, Michele and De Wolf, Ronald},
title = {Quantum search on bounded-error inputs},
year = {2003},
isbn = {3540404937},
publisher = {Springer-Verlag},
address = {Berlin, Heidelberg},
booktitle = {Proceedings of the 30th International Conference on Automata, Languages and Programming},
pages = {291–299},
numpages = {9},
location = {Eindhoven, The Netherlands},
series = {ICALP'03}
}

\end{document}